\documentclass[conference,letterpaper]{IEEEtran}
\usepackage[utf8]{inputenc} 
\usepackage[T1]{fontenc}
\usepackage{url}
\usepackage{ifthen}
\usepackage{cite}
\usepackage[cmex10]{amsmath} 

\usepackage[table]{xcolor}
\usepackage{algorithmic, algorithm}
\usepackage[english]{babel}
\usepackage{amssymb}
\usepackage{xcolor}
\usepackage{cite}
\usepackage{graphicx}
\usepackage{amsthm}
\usepackage{dirtytalk}
\usepackage[hidelinks]{hyperref}
\usepackage{tikz}
\usetikzlibrary{positioning}
\usepackage{orcidlink}
\usepackage[hidelinks]{hyperref}
\usepackage{amsmath, epsfig}
\usepackage{cite}
\usepackage{hyperref}
\usepackage{cleveref}
\usepackage{cancel}
\usepackage{comment}
\crefname{lemma}{Lemma}{Lemmas}
\crefname{figure}{Figure}{Figures}
\crefname{problem}{Problem}{Problems}
\crefname{proposition}{Proposition}{Propositions}
\crefname{theorem}{Theorem}{Theorems}
\crefname{remark}{Remark}{Remarks}
\crefname{example}{Example}{Examples}
\crefname{table}{Table}{Tables}
\crefname{corollary}{Corollary}{Corollaries}

\usepackage[a4paper, total={6.9in, 9.9in}]{geometry}
\normalsize

\theoremstyle{definition}
\newtheorem{theorem}{Theorem}
\newtheorem{lemma}{Lemma}
\newtheorem{definition}{Definition}
\newtheorem{corollary}{Corollary}

\newtheorem{example}{Example}
\newtheorem{remark}{Remark}

\newtheorem{conjecture}{Conjecture}

\newtheorem*{remark*}{Remark}
\newtheorem*{theorem*}{Theorem}
\newtheorem*{problem*}{Problem}

\newcommand{\ceil}[1]{\left\lceil {#1} \right\rceil}
\newcommand{\floor}[1]{\left\lfloor {#1} \right\rfloor}
\DeclareMathOperator{\Vol}{Vol}

\newcommand{\E}{\mathbb{E}}

\newcommand{\cA}{{\cal A}}

\title{The Labeled Coupon Collector Problem}

 \author{\IEEEauthorblockN{Andrew Tan\IEEEauthorrefmark{1}, Oriel Limor\IEEEauthorrefmark{4}, Daniella Bar-Lev\IEEEauthorrefmark{1}, Ryan Gabrys \IEEEauthorrefmark{3}, Zohar Yakhini\IEEEauthorrefmark{2}\IEEEauthorrefmark{4}, and Paul H. Siegel\IEEEauthorrefmark{1}}
\thanks{{The work of O. Limor and Z. Yakhini was funded by the European Union (DiDAX, 101115134). Views and opinions expressed are however those of the author(s) only and do not necessarily reflect those of the European Union or the European Research Council Executive Agency. Neither the European Union nor the granting authority can be held responsible for them. The work of D. Bar-Lev was supported in part by Schmidt Sciences. The work of D. Bar-Lev, R. Gabrys, A. Tan, and P. H. Siegel was supported in part by NSF Grant CCF-2212437}}
\IEEEauthorblockA{\IEEEauthorrefmark{1}\textit{Center for Memory and Recording Research, ECE Department, University of California, San Diego}} \IEEEauthorblockA{\IEEEauthorrefmark{2}\textit{School of Computer Science, Reichman University, Herzliya, Israel}} \IEEEauthorblockA{\IEEEauthorrefmark{3}\textit{Qualcomm Institute, Calit2, University of California, San Diego}} \IEEEauthorblockA{\IEEEauthorrefmark{4}\textit{Faculty of Computer Science, Technion -- Israel Institute of Technology, Haifa, Israel}} Email: \{a2tan, dbarlev, rgabrys, psiegel\}@ucsd.edu,   oriel.limor@cs.technion.ac.il, zohar.yakhini@runi.ac.il }

\IEEEoverridecommandlockouts
\IEEEaftertitletext{\vspace{-5pt}}

\begin{document}

\maketitle

\begin{abstract}


We generalize the well-known Coupon Collector Problem (CCP) in combinatorics. Our problem is to find the minimum and expected number of draws, with replacement, required to recover $n$ distinctly labeled coupons, with each draw consisting of a random subset of $k$ different coupons and a random ordering of their associated labels. We specify two variations of the problem, Type-I in which the set of labels is known at the start, and Type-II in which the set of labels is unknown at the start. We show that our problem can be viewed as an extension of the separating system problem introduced by R\'enyi and Katona, provide a full characterization of the minimum, and provide a numerical approach to finding the expectation using a Markov chain model, with special attention given to the case where two coupons are drawn at a time.

\end{abstract}

\section{Introduction}
\label{sec:introduction}

The Coupon Collector Problem (CCP) is a well-known combinatorial problem with broad applications. In the classical setting, there are $n$ distinct coupons, and the problem is to find the expected number of draws needed to collect all $n$ coupons, where in each draw one coupon is drawn independently and uniformly at random with replacement. The answer can be shown to be $nH_n$, where $H_n :=  \sum\limits_{j=1}^n \frac{1}{j}$ denotes the $n$-th harmonic number~\cite{Feller}.
A notable extension of the CCP is the group drawing scenario, where $k$ different coupons are drawn at a time for some fixed constant $k < n$ \cite{BT22, JS10, Stadje, AR01}. For $k = 2$, it was shown that the expected number of draws is  $\frac{1}{2} n H_n + O(1)$.

In this paper, we introduce a novel extension of the group drawing scenario, where each of the $n$ coupons is associated with a distinct label, and each draw, with replacement, yields a random subset of $k$ different coupons as well as a random ordering of their corresponding labels. The collector's goal is to correctly identify the label associated with each coupon, and so we call this extension the \textit{Labeled Coupon Collector Problem} (\textit{LCCP}). We consider two variations of the problem: \textit{Type-I}, where the collector knows the set of labels in advance, and \textit{Type-II}, where the collector does not. 

As an example, consider the LCCP with four coupons $\{1,2,3,4\}$ and corresponding labels $\{A,B,C,D\}$, and with $k=2$. One draw might return coupons $\{1, 2\}$ with labels $\{A, B\}$. Then, a subsequent draw returning coupons $\{2,3\}$ and labels $\{B,C\}$ would allow the collector to deduce the correspondences $1 \leftrightarrow A, 2 \leftrightarrow B, 3 \leftrightarrow C$. Moreover, in the Type-I setting, the collector could further conclude via elimination that $4\leftrightarrow D$. Thus, the LCCP generalizes the classical CCP and the group drawing extension, combining elements of both collection of and deduction from unordered sets, and presents potential applications in a variety of fields, such as network communications \cite{Li,Lee}, cryptography \cite{Goyal,Beunardeau}, group testing \cite{Alridge,Chan}, and data storage \cite{Lenz}.


Our study aims to analyze the minimum\footnote{The minimum number of draws for the LCCP is not trivial, unlike for the CCP ($n$ draws) and group drawing scenario ($\ceil{\frac{n}{k}}$ draws).} and the expected number of draws in the LCCP. A closely related problem on separating systems for finite sets was first proposed and studied by R\'enyi~\cite{Renyi61,Renyi62} and Katona~\cite{Katona}. In particular,~\cite{Katona} reduces the combinatorial problem of finding the ${\text{Type-I}}$ minimum to an analytical problem, yielding upper and lower bounds on the minimum; these bounds were later improved by Wegener~\cite{Wegener}. In our work, we provide a more complete characterization of the minimum number of draws as a function of $n$ and $k$ and address the problem of finding the expected number of draws. Finally, we note that a recent parallel work~\cite{BYYB2025} extends the study of the LCCP to \textit{partial recovery}, where the goal is to determine the labels either of a specific subset of coupons or of any subset of coupons of a specified size, and to the \textit{heterogeneous setting}, where the size of the draws varies randomly. 

The remainder of this paper is organized as follows. In Section II, we introduce basic definitions and formally define the LCCP, and establish a key connection between the LCCP and a problem on separating systems. In Section~III, we find the minimum number of draws for any $n$ and $k$. In Section IV, we define and use a Markov chain model to find the expected number of draws, with special attention paid to the case where $k=2$. Lastly, in Section V, we discuss our findings and draw immediate conclusions.

\section{Preliminaries}

In this section, we define the LCCP and establish the connection to separating systems ~\cite{Katona,Renyi61,Renyi62}.

\vspace{10pt}

\noindent
{\bf Problem [$(n,k)$-LCCP].}
Let $[n]:= \{1,2,...,n\}$ represent a set of coupons, $L$ be a set of labels with $|L| = n$, and $f: [n] \rightarrow L$ be a bijective \textit{labeling function} that assigns to each coupon, $j\in [n]$, a unique label, $f(j) \in L$. For a fixed drawing size $k \in \{1,...,n-1\}$, let $S_i\subseteq[n]$ represent the subset of coupons drawn in the $i$-th trial, whose labels are given by $f(S_i) = \{f(j): j\in S_i\}$, where $|S_i| = k$.
Define a random variable $Q(n,k)$ as the \textit{first} trial in which $f$ can be determined, i.e., all coupons' labels can be \textit{uniquely determined} given the system of pairs $\{(S_i,f(S_i))\}_{i=1}^{Q(n,k)}$, where $S_i$ and $f(S_i)$ are both \textit{unordered} $k$-subsets for all $i$. 

Find the \textit{minimum} and \textit{expectation} of $Q(n,k)$.\hfill$\square$

\smallskip
We consider this problem in two settings: one where $L$ is known at the start,  and one where $L$ is unknown; we refer to the corresponding LCCP problems as \textit{Type-I} and \textit{Type-II} respectively. The corresponding solutions for the minimum and expectation are denoted by $q_\mathrm{I}(n,k)$ and $\mathbb{E}[Q_{\mathrm{I}}(n,k)]$,  and 
$q_\mathrm{II}(n,k)$ and $\mathbb{E}[Q_{\mathrm{II}}(n,k)]$, respectively.

Next, we define \textit{separability} and \textit{coverability}.

\begin{definition}
    Given $n \in \mathbb{Z}^+$ and subsets $\{S_i\}_{i=1}^m \subseteq [n]$.
    \begin{itemize}
        \item A coupon $j \in [n]$ is \textit{separable} if for every $j' \neq j$, there exists $i \in [m]$ such that either $j \in S_i$ or $j' \in S_i$, but not both. Moreover, $\{S_i\}_{i=1}^m$ is a \textit{separating system} (or separates $[n]$) if every coupon is separable.
        \item A coupon $j \in [n]$ is \textit{covered} if $j \in \bigcup\limits_{i=1}^m S_i$. Moreover, $\{S_i\}_{i=1}^m$ is a \textit{covering system} (or \textit{covers} $[n]$) if every coupon is covered by $\{S_i\}_{i=1}^m$, i.e., $\bigcup\limits_{i=1}^m S_i =[n]$. \hfill $\square$
    \end{itemize}
\end{definition}

\noindent Note that the definition of a separating system is identical to the one in \cite{Katona}, and in the following lemma, we establish an intuitive characterization of separability and separating systems.  For subset $S \subseteq [n]$, let $S^\mathrm{c}$ denote its complement.

\begin{lemma}
    \label{lemma_separable}
    Given $n  \in \mathbb{Z}^+$ and subsets $\{S_i\}_{i=1}^m$. For all $i \in [m]$ and $j\in [n]$,  define 
     \begin{equation}
     \label{canonical_form}
      C_{i,j} = \begin{cases}
            S_i , & j \in S_i \\
            S_i^\mathrm{c}, & j \notin S_i
        \end{cases} 
        \end{equation}
      and
      \begin{equation}
      \label{candidates}
     c(j) = \bigcap\limits_{i=1}^m C_{i,j}. 
    \end{equation}
    Then, any $j \in [n]$ is separable given $\{S_i\}_{i=1}^m$ if and only if $\{j\} = c(j)$. Consequently, $\{S_i\}_{i=1}^m$ is a separating system if and only if this condition holds for all $j \in [n]$.
\end{lemma}

\begin{proof}
Observe that $j \in C_{i,j}$ for all $i\in [m], j\in [n]$, so $\{j\} \subseteq c(j)$. Equality holds if and only if for all $j' \neq j$ there exists $i \in [m]$ such that $j' \notin C_{i,j}$. By definition, this occurs for all $j'\neq j$ if and only if $j$ is separable.
\end{proof}
\begin{remark}
\label{remark_sep}
Note that $c(j)$ is the smallest subset containing~$j$ that can be generated via set operations on the sets $\{S_i\}$, so we call $c(j)$ the \textit{separation class} of $j$, i.e., the set of coupons that cannot be separated from $j$ given $\{S_i\}_{i=1}^m$. 
$\square$

\end{remark}

\begin{theorem}
\label{thm_deducing} 
    For any $n \in \mathbb{Z}^+$, $\{S_i\}_{i=1}^m \subseteq [n]$, the label of every coupon can be determined given $\{(S_i, f(S_i))\}_{i=1}^m$ if and only if $\{S_i\}_{i=1}^m$ is a separating system (additionally a covering system in the Type-II setting).
\end{theorem}

\begin{proof}
    In the Type-II setting, covering is necessary, since each label must be observed at least once, that is, $\bigcup\limits_{i=1}^m f(S_i) = L = f([n]) $, so $ \bigcup\limits_{i=1}^m S_i =[n]$ by bijectivity. 
    
    Next, for any $j \in [n]$, $f(j)$ can be determined for any $j \in [n]$ only if $\{f(j)\}$ can be formed via set operations on $\{f(S_i)\}_{i=1}^m$. The smallest set of labels that can be obtained in this manner is $c(f(j)) =\bigcap\limits_{i=1}^m D_{i,j}$, where 
    \vspace{-0.75em}
    $$D_{i,j} = \begin{cases}
        f(S_i), & f(j) \in f(S_i) \\
        f(S_i)^\mathrm{c}, & f(j) \notin f(S_i).
    \end{cases}$$ 
    By bijectivity, $\{f(j)\} {=} c(f(j)) {=} f(c(j))$, so $c(j) {=} \{j\}$. By \cref{lemma_separable}, $\{S_i\}_{i=1}^m$ must separate $[n]$ in both settings.

    Conversely, if $\{S_i\}_{i=1}^m$ separates $[n]$ (and covers $[n]$ in the Type-II setting), then by \cref{lemma_separable}, the label of coupon $j \in [n]$ is the sole element of $f(c(j))$. 
\end{proof}
\begin{remark}
\label{remark_deducing}
    \cref{thm_deducing} implies that number of draws required to deduce the label of every coupon is the number of subsets required to form a separating system (additionally a covering system in the Type-II setting). Moreover, $c(j)$ is also the set of coupons whose label could be \textit{confused} with that of $j$ given $\{(S_i,f(S_i))\}_{i=1}^m$. \hfill $\square$
\end{remark}

Finally, we introduce a matrix representation used in~\cite{Katona}, where a sequence of drawn subsets $\{S_i\}_{i=1}^m\subseteq[n]$ corresponds to a \textit{draw matrix} $M \in \mathbb{F}_2^{m\times n}$, such that $M_{ij} = 1$ when $j \in S_i$. Thus, each row of $M$ has weight $k$, and each column corresponds to a coupon. Moreover, it can be observed that $\{S_i\}_{i=1}^m$ separates $[n]$ if and only if all of the columns of $M$ are distinct, and covers $[n]$ if and only if all of the columns are nonzero (i.e., have nonzero weight). 

\section{Minimum Number of Draws} 

In this section, we find the minimum number of draws for the $(n,k)$-LCCP in both Type-I and Type-II settings. 
We will use the fact that if there exists $M \in \mathbb{F}_2^{m\times n}$ with weight-$k$ rows and distinct (resp. distinct and nonzero) columns, then $q_\mathrm{I}(n,k) \leq m$ (resp. $q_\mathrm{II}(n,k) \leq m$).

\subsection{Fundamental Results}

We begin with several universal lower bounds.

\begin{lemma}
    \label{lemma_logbound}
    For all $n, k \in \mathbb{Z}^+$ with $n > k$:
    \begin{enumerate}
        \item [(a)]\label{lemma_logbound_a} $q_{\mathrm{I}}(n,k) \geq \ceil{\log_2 n}, q_{\mathrm{II}}(n,k) \geq \ceil{\log_2(n+1)}$.
        \item [(b)]\label{lemma_logbound_b} $q_{\mathrm{I}}(n,k), q_{\mathrm{II}}(n,k) \geq \ceil{\log_2 2k}$.
        \item [(c)]\label{lemma_logbound_c} The bounds in (a) and (b) are identical and met with equality in the Type-I (resp. Type-II) setting when ${n = 2k}$ (resp. $n=2k-1$).
        \item [(d)]\label{lemma_logbound_d} If $k >1$, then $q_{\mathrm{I}}(2k,k) = q_{\mathrm{I}}(2k-1,k)$.
    \end{enumerate}
\end{lemma}
\begin{proof}
Suppose $\{S_i\}_{i=1}^m$ separates $[n]$ and corresponds to $M\in \mathbb{F}_2^{m\times n}$. The columns of $M$ form a subset of $\mathbb{F}_2^m$, so $n \leq |\mathbb{F}_2^m| = 2^m$, yielding the bounds in (a) (for Type-II, exclude the all-zero vector). Moreover, all vectors in $\mathbb{F}_2^m$ collectively have $2^{m-1}$ ones in each entry, so $k \leq 2^{m-1}$, yielding the bound in (b).

To prove (c), we construct $M \in \mathbb{F}_2^{m\times n}$ for $m = \ceil{\log_2 2k}$ corresponding to a separating system for $[2k]$. Let the columns of $M$ consist of the all-ones vector, the all-zeros vector, and 
$k-1$  distinct pairs of nonzero vectors that are complements of each other, which exist because $k-1 \leq 2^{m-1}-1$. This yields a separating system for $[2k]$; a separating and covering system for $[2k-1]$ is obtained by deleting the all-zeros column.

Finally, (d) holds since $\ceil{\log_2 2k} = \ceil{\log_2 (2k-1)}$ when $k > 1$. 
\end{proof}

\begin{remark}
    When $k$ is a power of $2$, the draw matrix 
  for the Type-II setting constructed in  Lemma~\ref{lemma_logbound}, part (c) is a Hamming matrix. \hfill $\square$
\end{remark}

From \cref{thm_deducing}, it is clear that $q_{\mathrm{II}}(n,k) \geq q_{\mathrm{I}}(n,k)$, though the next lemma fully characterizes the relation between the Type-I and Type-II minima.

\begin{lemma}
    \label{lemma_type12_relations}
    For any $k\in \mathbb{Z}^+$:
    \begin{equation}
        q_\mathrm{II}(n,k) = \begin{cases}
            q_\mathrm{I}(n,k) , & k < n < 2k \\
            q_\mathrm{I}(n+1,k), & n \geq 2k-1.
        \end{cases}
    \end{equation}
    Moreover, when $n \geq 2k-1$, $q_\mathrm{II}(n,k)$ is either $q_{\mathrm{I}}(n,k)$ or $q_{\mathrm{I}}(n,k) + 1$.
\end{lemma}

\begin{proof}
    See Appendix. 
\end{proof}

    An immediate consequence of \cref{lemma_type12_relations} is the following monotonicity property.
\begin{lemma}
\label{lemma_ordering} When $n \geq 2k-1$, $q_{\mathrm{I}}(n+1,k) \geq q_{\mathrm{I}}(n,k)$ and $q_{\mathrm{II}}(n+1,k) \geq q_{\mathrm{II}}(n,k)$. \hfill $\square$
\end{lemma}

\begin{proof}
By \cref{lemma_type12_relations}, $q_{\mathrm{I}}(n+1,k) =q_{\mathrm{II}}(n,k) \geq q_{\mathrm{I}}(n,k)$, and $q_{\mathrm{II}}(n+1,k) = q_{\mathrm{I}}(n+2,k) \geq q_\mathrm{I}(n+1,k) = q_{\mathrm{II}}(n,k)$.

\end{proof}

We conclude this subsection with an easily proved symmetry property also observed in~\cite{Katona}.

\begin{lemma}
\label{lemma_symmetry}
     $q_{\mathrm{I}}(n,k) = q_{\mathrm{I}}(n,n-k)$. \hfill $\square $
\end{lemma}

\begin{proof}
    Take any matrix $M \in \mathbb{F}_2^{q_\mathrm{I}(n,k) \times n}$ with weight-$k$ rows and distinct columns; its component-wise complement $M^\mathrm{c}$ will have weight-$(n-k)$ rows and distinct columns, so $q_{\mathrm{I}}(n,n-k) \leq q_{\mathrm{I}}(n,k)$. Applying the same argument in the other direction yields the result.
\end{proof}

\begin{remark}
    It follows from \cref{lemma_type12_relations,lemma_symmetry} that the scope of the LCCP can be reduced to the regime $n \geq 2k$, since if $n < 2k$, then $n > 2(n-k)$. An analog of \cref{lemma_symmetry} for the distribution of $Q_\mathrm{I}(n,k)$ was shown in \cite{BYYB2025}. \hfill $\square$
\end{remark}

\subsection{Full Characterization}

We begin with a lower bound that is useful when $n$ is large with respect to $2k$.

\begin{lemma} 
\label{lemma_fixedk}
     $q_{\mathrm{I}}(n,k) \geq \ceil{\frac{2n-2}{k+1}}$ and $q_{\mathrm{II}}(n,k)\geq \ceil{\frac{2n}{k+1}}$.
    
\end{lemma}

\begin{proof}
    Consider any separating 
    system $\{S_i\}_{i=1}^m$ for some $m \in \mathbb{Z}^+$ corresponding to draw matrix $M \in \mathbb{F}_2^{m \times n}$. For $w\in \{0,...,m\}$, let $A_w$ be the number of weight-$w$ columns, so $\sum\limits_{w=0}^m A_w = n$ and $A_w \leq \binom{m}{w}$ (since columns are distinct). Then, the total number of ones in $M$ is $$mk = \sum\limits_{w=0}^m A_w w = A_1 + \sum\limits_{w=2}^m A_w w \geq A_1 + 2(n-A_0-A_1),$$ implying  $2n \leq mk + 2A_0 + A_1$. 
    
   The bound for Type-I follows from $A_0 \leq 1, A_1 \leq m$, and for Type-II from  $A_0 = 0, A_1 \leq m$.
\end{proof}

\begin{remark}
    The bound in Lemma~\ref{lemma_fixedk}  is tight when $n \geq \binom{k+1}{2}$ (see \cref{lemma_ceil}), whereas the bound in \cref{lemma_logbound} is loose in this regime, as  illustrated by \cref{ex: opt min}. \hfill $\square$
\end{remark}

\begin{example}\label{ex: opt min}
    Consider the Type-II $(15,5)$-LCCP and subsets $\{S_i\}_{i=1}^5 \subseteq [15]$ shown in \autoref{fig:const_min}, where row $i$ corresponds to set $S_i$ and column $j$ corresponds to coupon $j$. It can be seen that $15 = \binom{5+1}{2}$, and that $\{S_i\}_{i=1}^5$ separates and covers $[15]$, achieving the lower bound  $\frac{2n}{k+1} = 5$. \hfill $\square$
    \begin{figure}[ht]
    \centering \includegraphics[width=0.85\linewidth]{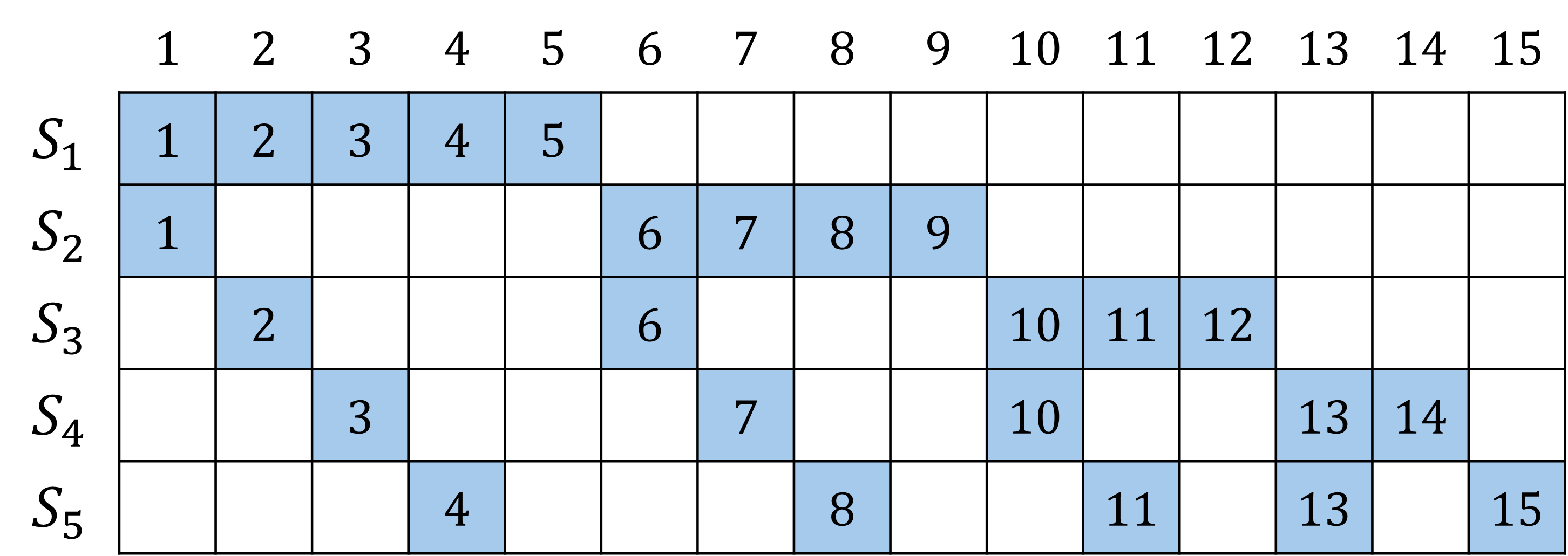}
    \caption{Visualization of a separating and covering system achieving the minimum of 5 draws for the Type-II $(15,5)$-LCCP } 
    \label{fig:const_min}
\end{figure}
\end{example}

The proof of \cref{lemma_fixedk} suggests that it is worth considering the dual problem: namely, given a fixed number of draws $m$ and drawing size $k$, what is the \textit{maximum} number of coupons whose label can be determined in $m$ draws? This can be achieved by using low column weights in the draw matrix, such as in \cref{ex: opt min} where each column has a weight of $1$ or $2$. The analysis of such matrices is facilitated by introducing the following concept.

\begin{definition}
   Denote the volume of a Hamming ball in $\mathbb{F}_2^m$ with radius $r$ as $\Vol_2(r,m) = \sum\limits_{\ell=0}^r \binom{m}{\ell}$. For $n \in \mathbb{Z}^+$ such that $n \leq 2^m$, define its \textit{minimum Hamming radius} $r_m(n)$ as the smallest value of $r$ such that $\Vol_2(r,m) \geq n$; i.e., $r_m(n) = \min\{ r : \Vol(r,m) \geq n\}$ \hfill $\square$
\end{definition}

\begin{lemma}
\label{lemma_efficiency}
    For $k \in \mathbb{Z}^+$, define the sequence $\{m_{k,\ell}\}_{\ell=1}^{\infty} = \{ \ceil{\log_2 2k}, \ceil{\log_2 2k}+1,... \}$, and, for all $\ell \in \mathbb{Z}^+$, define 
    \begin{align*}
    \hat{k}_\ell &= \Vol_2(r_{m_{k,\ell}{-}1}(k){-} 1, m_{k,\ell}{-}1)\\
        n_{k,\ell} &= \Vol_2 (r_{m_{k,\ell}{-}1}(k),m_{k,\ell}){-}1 + \floor{\frac{m_{k,\ell}(k{-}\hat{k}_\ell)}{r_{m_{k,\ell}-1}(k){+}1} }.
    \end{align*}

Then, $\{n_{k,\ell}\}_{\ell=1}^\infty$ is a strictly increasing sequence, and for all $\ell\in \mathbb{Z}^+$, $n_{k,\ell} \geq 2k-1$ and $q_\mathrm{I}(n_{k,\ell}+1,k) = q_{\mathrm{II}}(n_{k,\ell},k)=m_{k,\ell}$.
\end{lemma} 

\begin{proof}

 We prove the result for $q_{\mathrm{II}}(n_{k,\ell},k)$. The result for  $q_{\mathrm{I}}(n_{k,\ell},k)$ will follow immediately from \cref{lemma_type12_relations}, since it will be shown that $n_{k, \ell} \geq 2k-1$. 
    
    Note that $r_{m_{k,\ell}-1}(k) \leq m_{k,\ell}-1$, and if $r_{m_{k,\ell}}(k) = m_{k,\ell}-1$, then $k > \Vol_2(m_{k,\ell}-2,m_{k,\ell}-1) = 2^{m_{k,\ell}-1} - 1$ and $k \leq \Vol_2(m_{k,\ell}-1,m_{k,\ell}-1) = 2^{m_{k,\ell}-1}$, implying ${k = 2^{m_{k,\ell}-1}}$, and a draw matrix $M$ whose columns consist of all distinct nonzero vectors in $\mathbb{F}_2^{m_{k,\ell}}$ represents a separating and covering system for $[n_{k,\ell}]$. Lemma~\ref{lemma_logbound} part (c) then gives the desired result. Thus, we assume that $r_{m_{k,\ell}}(k) \leq m_{k,\ell}-2$. 

    By Theorem 1 in \cite{Katona}, there exists $M \in \mathbb{F}_2^{m\times n}$ corresponding to a separating and covering system for $[n]$ if and only if there exist integers $A_1,...,A_m \in \mathbb{Z}^+$ such that  
    \begin{align}
    \label{eqn_valid}
        n& = \sum\limits_{w=1}^m A_w\nonumber\\  A_w &\leq \binom{m}{w} \; \forall \; w=1,...,m
        \\mk &= \sum\limits_{w=1}^m w A_w \nonumber .
    \end{align}
    
    Thus, it suffices to show that a solution to \eqref{eqn_valid} with ${n = n_{k,\ell}}$ exists when $m = m_{k,\ell}$ but not when $m = m_{k,\ell} -1$. 

    Set $m = m_{k,\ell}$, $r = r_{m-1}(k)$, and $\hat{k} = \hat{k}_\ell  = \Vol_2 (r-1,m-1)$, and consider the following maximization problem:

    \begin{equation}
    \label{eqn_max}
\max\limits_{\substack{ A_w \leq \binom{m}{w}\\\sum\limits_{w=1}^m wA_w \leq mk}} \; \sum\limits_{w=1}^m A_w. 
    \end{equation}
    
   Expressed in terms of the well-known knapsack problem, the goal here is to maximize the number of items carried, given a total weight limit of $mk$  and a supply of $\binom{m}{w}$ items of weight $w$, for $w = 1,...,m$. Note that the maximum of~\eqref{eqn_max} is an upper bound to the maximum value of $n$ where a solution to $\eqref{eqn_valid}$ exists.
    
    Observe that one maximizer for \eqref{eqn_max} is:
    \begin{equation}
    \label{eqn_maximizer}
        A_w = \begin{cases}
            \binom{m}{w}, & 1 \leq w \leq r  \\
             \floor{\frac{m(k- \hat{k})}{r+1}}, & w = r+1 \\
             0, & \mathrm{otherwise}.
        \end{cases}
    \end{equation}
\noindent
    Indeed, since $w \binom{m}{w} = m\binom{m-1}{w-1}$, then $\sum\limits_{w=1}^m wA_w = \sum\limits_{w=1}^r w \binom{m}{w} + (r+1)A_{r+1} = \sum\limits_{w=1}^r m\binom{m-1}{w-1} + (r+1)A_{r+1} = m \hat{k} + (r+1)A_{r+1}$, and so $mk - \sum\limits_{w=1}^m wA_w \leq r$, by definition of $A_{r+1}$. No more items can be added, and no item can be replaced by two other items without violating the constraints. Moreover, the total number of items is $\sum\limits_{w=1}^m A_w = \sum\limits_{w=1}^r \binom{m}{w} + A_{r+1} = \Vol_2 (r,m) -1 + A_{r+1} = n_{k,\ell}$.

    The $A_w$ values in \eqref{eqn_maximizer} do not necessarily satisfy $mk = \sum\limits_{w=1}^m wA_w$, but this can be  fixed by letting $d = mk - \sum\limits_{w=1}^m wA_w \leq r$, decreasing $A_{r+2-d}$ by $1$, and increasing $A_{r+2}$ by $1$. Hence, there exists a separating and covering system for $[n_{k,\ell}]$, so $q_{\mathrm{II}}(n_{k,\ell},k) \leq m_{k,\ell}$. Additionally, for all $\ell\in \mathbb{Z}^+$, $n_{k,\ell}$ is the maximum of \eqref{eqn_max}, and more importantly is the maximum $n$ for which $[n]$ can be separated and covered by $m_{k,\ell}$ draws of size $k$. Since $m_{k,\ell} \geq \ceil{\log_2 2k}$, $[2k-1]$ can be separated and covered by $m_{k,\ell}$ draws (via \cref{lemma_logbound} part (c)), so $n_{k,\ell } \geq 2k-1$ by maximality.

    To show equality, we show that no separating and covering system of $[n_{k,\ell}]$ exists for $m_{k,\ell}-1$ draws, namely by showing that $n_{k,\ell}$ cannot be achieved in \eqref{eqn_max} with $m_{k,\ell}$ replaced by $m_{k,\ell}-1$ without violating the total weight constraint. Moreover, since $n_{k,\ell-1}$ is the maximum $n$ for which $[n]$ can be separated and covered by $m_{k,\ell-1}$ draws, this will also imply that $n_{k,\ell-1} < n_{k,\ell}$, or in other words, $\{n_{k,\ell}\}_{\ell=1}^\infty$ is strictly increasing.

    We start with the earlier solution from \eqref{eqn_maximizer}, which now violates the supply constraint for each $w \leq r$, and remove the "overfill" in each weight class.  The total number of items removed is $\sum\limits_{w=1}^r \binom{m}{w}-\binom{m-1}{w} = \sum\limits_{w=1}^r \binom{m-1}{w-1} = \hat{k}$, and the total weight is decreased by $\sum\limits_{w=1}^r w \binom{m-1}{w-1}$. Since the total weight limit has been reduced by $k$, then we must add $\hat{k}$ items of weight at least $r+1$ so that their total weight is at most $d + \sum\limits_{w=1}^r w \binom{m-1}{w-1} -k$, where $d\leq r$. However, since $k > \hat{k}$ and $\hat{k} = \sum\limits_{w=0}^r \binom{m-1}{w} \geq r$, then $d + \sum\limits_{w=1}^r w \binom{m-1}{w-1} -k< r + r\hat{k} -\hat{k} = r+(r-1)\hat{k} \leq r\hat{k} <(r+1) \hat{k}$, which means that the total weight limit cannot be satisfied with $n_{k,\ell}$ items. Therefore, $q_{\mathrm{II}}(n_{k,\ell},k) = m_{k,\ell}$.
\end{proof}
\begin{remark}
   Note that $n_{k,\ell}$  can be viewed as ``thresholds'' determining the maximum number of coupons whose labels can be deduced with $m_{k,\ell}$ draws of size $k$. This is evident later in \cref{thm_minimum_full}. \hfill $\square$
\end{remark}
The following result is a consequence of \cref{lemma_type12_relations,lemma_ordering,lemma_fixedk,lemma_efficiency}. It was first proved by other means in \cite{Katona}.
 \begin{lemma}
 \label{lemma_ceil} Lemma~\ref{lemma_fixedk} holds with equality if $n \geq \binom{k+1}{2}$.
\end{lemma}
\begin{proof}
   Consider the Type-II bound. Setting $m = \ceil{\frac{2n}{k+1}}$, observe that $\Vol_2(1,m-1) =  m \geq \frac{2\binom{k+1}{2}}{k+1} = k$, so $r_{m-1}(k) = 1$. By \cref{lemma_efficiency}, $q_\mathrm{II}(n', k) = m$, where $n' = \Vol_2(1,m)-1 + \floor{\frac{m(k-\Vol_2(0,m-1))}{1+1}} =  m + \floor{\frac{m(k-1)}{2}}$, which can also be written as $n' = m-1 + \ceil{\frac{m(k-1)+1}{2}} = \ceil{\frac{m(k+1)-1}{2}}$. Moreover, $m = \ceil{\frac{2n}{k+1}} \geq \frac{2n}{k+1}$, so $n' \geq \ceil{\frac{\frac{2n}{k+1}(k+1) -1}{2}} = \ceil{n-\frac{1}{2}}=n$. By \cref{lemma_ordering}, $q_\mathrm{II}(n,k) \leq  m$. 

Finally, by \cref{lemma_fixedk}, $q_{\mathrm{II}}(n,k) = m$.  The respective expression for Type-I follows immediately from \cref{lemma_type12_relations} since $n \geq \binom{k+1}{2}\geq2k-1$.
\end{proof}

\begin{corollary}
    $q_{\mathrm{I}}(\binom{k+1}{2}p+1,k) = q_{\mathrm{II}}(\binom{k+1}{2}p,k) = kp$ \hfill $\square$ for any $p \in \mathbb{Z}^+$.
\end{corollary}

\begin{corollary}
    $q_{\mathrm{I}}(n,2)= \ceil{\frac{2n-2}{3}}, \; q_{\mathrm{II}}(n,2) = \ceil{\frac{2n}{3}}$ \hfill $\square$
\end{corollary}

\cref{thm_minimum_full} fully characterizes $q_{\mathrm{I}}(n,k)$ and $q_{\mathrm{II}}(n,k)$. 

\begin{theorem}
    \label{thm_minimum_full}
    Given any $k \in \mathbb{Z}^+$, define $m_{k,\ell}$ and $n_{k,\ell}$ for $\ell \in \mathbb{Z}^+$ as in \cref{lemma_efficiency}.
    \begin{enumerate}
        \item[(a)] If $n \geq \binom{k+1}{2}$, then $q_{\mathrm{I}}(n+1,k) = q_\mathrm{II} (n,k) = \ceil{\frac{2n}{k+1}}$.
        
        \item[(b)] If $2k-1 \leq n_{k,\ell-1} < n\leq n_{k,\ell} \leq \binom{k+1}{2}$ for some $\ell$, then $q_\mathrm{I}(n+1,k) =q_{\mathrm{II}} (n,k) =   m_{k,\ell}$.
        \item [(c)] If $k < n<2k$, then $q_\mathrm{I}(n,k) {=} q_\mathrm{II}(n,k) {=} q_\mathrm{I}(n,n-k)$. 
    \end{enumerate}
\end{theorem}
\begin{proof}
    (a) follows from \cref{lemma_ceil}, and (c) follows from \cref{lemma_type12_relations,lemma_symmetry}. If $n$ falls in (b), then by the proof of \cref{lemma_efficiency}, 
    $q_\mathrm{II}(n,k) > m_{k,\ell-1} = m_{k,\ell}-1$, 
  and by \cref{lemma_ordering},  $q_\mathrm{II}(n,k) \leq q_{\mathrm{II}}(n_{k,\ell},k) = m_{k,\ell}$, so $q_\mathrm{I}(n+1,k) =q_{\mathrm{II}} (n,k) =   m_{k,\ell}$.
    \end{proof}
\begin{remark}
\label{remark_min}
    Theorem~\ref{thm_minimum_full} can be used to compute $q_\mathrm{I}(n,k)$ and $q_\mathrm{II}(n,k)$ for any $n,k$.
   Ongoing work is directed at simplifying the computation of the thresholds $n_{k,\ell}$.  Interestingly, \cite{Katona} hinted at the concept of minimum Hamming radius, and here we used it to fully characterize the minimum number of draws required for any $n,k$.
   \hfill $\square$
\end{remark}

\section{Expected Number of Draws}

In this section, we study the expected number of draws in the LCCP using a Markov chain model that is
\begin{itemize}
    \item \textit{absorbing}: there exists a set of absorbing states $\cA$ that cannot be left once entered;
    \item \textit{unidirectional}: no state can be re-entered.
\end{itemize}

In an absorbing Markov chain, the \textit{expected time to absorption} (ETA) of a transient (non-absorbing) state is the expected number of transitions to reach an absorbing state. Moreover, if $P = (p_{ij})$ is the transition probability matrix, then the ETA $\mu_i$ for state $i$  is $\mu_i = 1+ \sum\limits_{j \notin \cA} p_{ij} \mu_j$  and 
\begin{equation}
\label{eqn_expected}
    (I-P_{\mathrm{tran}})\mu = \textbf{1}
\end{equation}
where $\mu$ is the vector of ETA values, $P_{\mathrm{tran}}$ is the submatrix of $P$ restricted to transient states, and \textbf{1} is the all-ones vector~\cite{KS76}.

For the LCCP, we consider a unidirectional chain with exactly one initial state and one absorbing state.
The expected number of draws will be the ETA of the initial state, which can be found by solving \eqref{eqn_expected}. Note that here $P_{\mathrm{tran}}$ is triangularizable due to unidirectionality. 

In the classical CCP, states can represent the number of unique coupons that have been collected, and the resulting Markov chain is shown in \cref{markovccp} (omitting self-loops). The ETA of the initial state is known to be $n H_n$. 


 In the $(n,k)$-LCCP, the states record the number, $s_\ell$, of coupons with $\ell$ possible known labels. Referring to \cref{remark_deducing}, in the Type-I setting, $s_\ell$ is the number of coupons whose separation class has cardinality~$\ell$. In the Type-II setting, they must also have been drawn at least once.
\begin{figure}
    \centering
\includegraphics[width=0.8\linewidth]{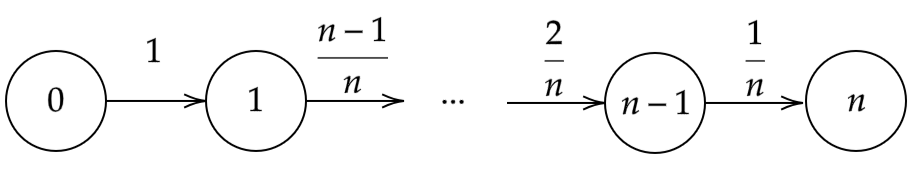}
    \caption{Markov chain graph (self-loops omitted) for the classical CCP. }
    \label{markovccp}
\end{figure}

Define a relation $\sim$ on $[n]$ such that $j_1 \sim j_2 $ if $j_1, j_2$ have been drawn and $c(j_1) = c(j_2)$, or if neither of $j_1, j_2$ has been drawn. Observe that $\sim$ is an equivalence relation, the separation classes are equivalence classes, and $[n]$ can be partitioned into disjoint separation classes and possibly a class of coupons that have not been drawn. Therefore, we represent each state as an integer $k$-tuple $(s_1,...,s_k)$ satisfying the linear Diophantine inequality:
\begin{align}
    \sum\limits_{\ell=1}^k s_\ell \leq n \\
    \ell | s_\ell \; \forall \; \ell\in \{1,2,...,k\},
\end{align}
where $s_\ell$ is the number of separation classes of size $\ell$. The only difference between the Type-I and Type-II settings is that the ``undrawn'' separation class can sometimes be counted among the $s_\ell$'s in the Type-I setting.


To construct the Markov chain, we iteratively find the states, starting  from the initial state $(0,0,...,0)$ and ending in the final state $(n,0,...,0)$, by determining all ways of drawing subsets and creating state transitions. Since subsets are drawn uniformly at random, the transition probabilities can be found by normalizing the number of ways to transition from one state to another by $\binom{n}{k}$. \cref{tpm} lists all possible state transitions for the $\text{Type-I}$ and Type-II problems in the $(n,2)$-LCCP, which can be used to construct the corresponding Markov chains. An example is given in \cref{markovlccp62} for the $(6,2)$-LCCP.

\begin{figure}[t]
    \centering
    \includegraphics[width=1\linewidth]{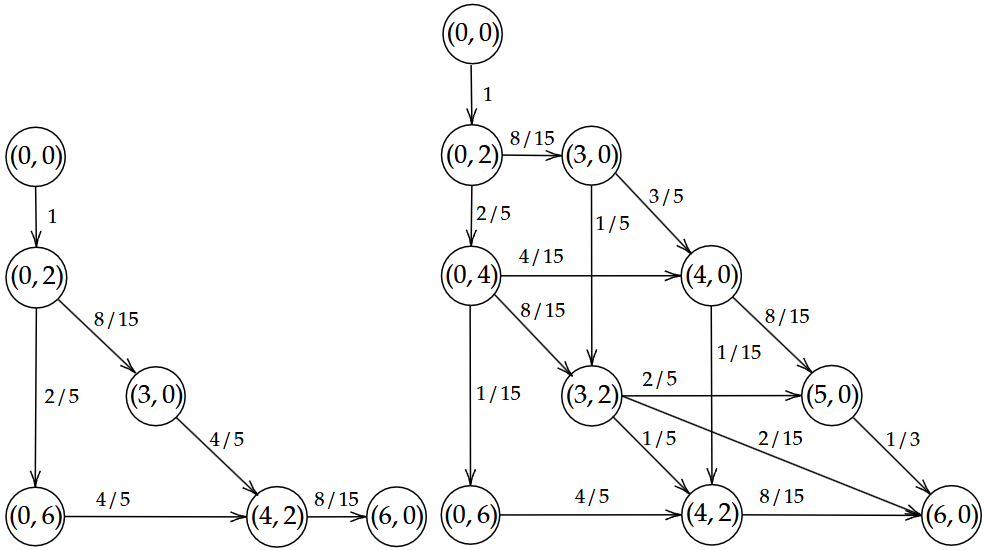}
    \caption{Markov chain models for the $(6,2)$-LCCP (left: Type-I, right: Type-II); self-loops are omitted.}
    \label{markovlccp62}
\end{figure}

\begin{table}[ht]  
    \caption{$(n,2)$-LCCP Markov chain state transitions \\(top: Type-I, bottom: Type-II)} 
    \centering
    \begin{tabular}{|c|c|c|} \hline 
         New State&  Requirements& No. of ways\\ \hline 
         $(s_1,s_2)$ & $s_1 \geq 2$ or $s_2 \geq 2$ & $\binom{s_1}{2} + \frac{s_2}{2}$ \\ \hline

         $(s_1 + 2, s_2 - 2)$ & $s_1 \geq 1, s_2 \geq 2$ & $s_1s_2$ \\ \hline 

         $(s_1 + 1, s_2)$ & $s_1 \geq 1, s_1 + s_2 < n-3$ & $s_1(n-s_1 - s_2)$ \\ \hline 

         $(s_1 + 4, s_2-4)$ & $s_2 \geq 4$ & $\frac{s_2(s_2-2)}{2}$ \\ \hline

         $(s_1 + 3, s_2 - 2)$ & $s_2 \geq 2, s_1 + s_2 < n-3$ & $ s_2(n-s_1-s_2)$\\ \hline

         $(s_1 + 3, s_2)$ & $s_2 \geq 2, s_1 + s_2 = n-3$ & $ s_2(n-s_1-s_2)$\\ \hline

         $(s_1, s_2 + 2) $ & $s_1+s_2 < n-4$ & $\binom{n-s_1-s_2}{2}$ \\ \hline 

         $(s_1, s_2 + 4) $ & $s_1+s_2 = n-4$ & $\binom{n-s_1-s_2}{2}$ \\ \hline 

         $(s_1+1, s_2 + 2)$ & $s_1\geq 1, s_1+s_2=n-3$ & $3(1+s_1)$ \\\hline 

    \end{tabular}
    
    \vspace{5pt}
    
    \begin{tabular}{|c|c|c|} \hline 
         New State&  Requirements& No. of ways\\ \hline 
         $(s_1,s_2)$ & $s_1 \geq 2$ or $s_2 \geq 2$ & $\binom{s_1}{2} + \frac{s_2}{2}$ \\ \hline

         $(s_1 + 2, s_2 - 2)$ & $s_1 \geq 1, s_2 \geq 2$ & $s_1s_2$ \\ \hline 

         $(s_1 + 1, s_2)$ & $s_1 \geq 1$ & $s_1(n-s_1 - s_2)$ \\ \hline 

         $(s_1 + 4, s_2-4)$ & $s_2 \geq 4$ & $\frac{s_2(s_2-2)}{2}$ \\ \hline

         $(s_1 + 3, s_2 - 2)$ & $s_2 \geq 2, s_1 + s_2 \leq n-1$ & $ s_2(n-s_1-s_2)$\\ \hline

         $(s_1, s_2 + 2) $ & $s_1+s_2 \leq n-2$ & $\binom{n-s_1-s_2}{2}$ \\ \hline        
    \end{tabular}
    \label{tpm}
    \vspace{2pt}
\end{table}

\begin{table} 
\caption{Expected number of draws for small $n$ when $k = 2$} 
\vspace{2pt}
    \centering
    \begin{tabular}{|c|c|c|c|c|} \hline 
         $n$&  $3$&  $4$&  $5$&  $6$\\ \hline 
         $\mathbb{E}[Q_{\mathrm{I}}(n,2)]$&  $\frac{5}{2}$&  $\frac{5}{2}$&  $\frac{47}{12}$&  $\frac{291}{56}$\\[0.5ex] 
         \hline 
         $\mathbb{E}[Q_{\mathrm{II}}(n,2)]$&  $\frac{5}{2}$&  $\frac{41}{10}$&  $\frac{121}{21}$&  $\frac{1817}{247}$\\ [0.5ex]
         \hline
    \end{tabular}
    \label{table_small_n}
    \vspace{2pt}
\end{table}


Finally, we present some results for the $(n,2)$-LCCP obtained using this approach. \cref{table_small_n} gives the numerically 
computed values of $\E[Q_\mathrm{I}(n,2)]$ and $\E[Q_\mathrm{II}(n,2)]$ for the first few values of $n$, and \cref{plot_expectations} gives normalized expectations $\frac{1}{n}\E[Q_\mathrm{I}(n,2)]$ and $\frac{1}{n}\E[Q_\mathrm{II}(n,2)]$ for selected values of $n$, alongside a plot of the sequence $\frac{1}{2}H_n$ as a reference.

\begin{figure}[t!]
    \centering
    \includegraphics[width=1\linewidth]{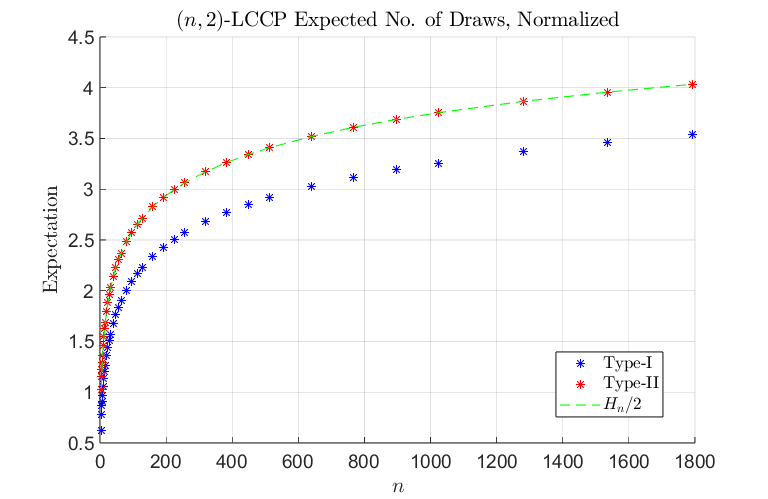}
    \caption{Expected number of draws for the $(n,2)$-LCCP computed via Markov chain analysis, normalized by $n$.}
    \label{plot_expectations}
\end{figure}

\cref{table_small_n} and \cref{plot_expectations} suggest that $\E[Q_\mathrm{II}(n,2)] \geq \E[Q_\mathrm{I}(n,2)] $, as well as 
the following conjectures. 
\begin{conjecture}\label{conj:exp-type-diff}
    $\E[Q_\mathrm{II}(n,2)] - \E[Q_\mathrm{I}(n,2)] = \frac{1}{2}n + o(n)$.
\end{conjecture}\label{conj:exp-type-two-2}
\vspace{-2pt}
\begin{conjecture} $\E[Q_\mathrm{II}(n,2)] = \frac{1}{2}nH_n +o(n)$. \hfill $\square$
\end{conjecture}

Namely, for $k = 2$, the difference between the Type-II and Type-I expectations may be approximately $\frac{1}{2}n$ for large $n$, and the Type-II expectation closely matches that of the $k=2$ group drawing scenario mentioned in Section~\ref{sec:introduction}. 
It is intriguing that the harmonic number $H_n$ appears to play a role in the expectation for the LCCP.


\begin{remark} In ongoing work, we seek to provide concentration bounds on the expectation  and address how it scales with respect to the number of coupons. 
\end{remark}
\vspace{2pt}
\section{Discussion}

We introduced the LCCP as an extension of the CCP and included two settings where the label set is known or unknown in advance. We connected the LCCP to R\'enyi and Katona's problem of minimal separating systems to characterize the minimum required number of draws in both settings. We then provided insight about the expected number of draws using a Markov chain model that can be further analyzed in future work to compute other key statistics and the distribution of the required number of draws. Several directions for future research were indicated.

\newpage

\newpage
\bibliographystyle{IEEEtran}

\newpage
\phantom{ }
\newpage

    \appendix
\noindent
\textbf{Proof of \cref{lemma_type12_relations}}:
    Suppose $n < 2k$. Let $m = q_{\mathrm{I}}(n,k)$ and $S_1,...,S_m \subseteq [n]$ be a separating system corresponding to $M \in \mathbb{F}_2^{m\times n}$, so each row of $M$ has weight $k$, and the columns are distinct. If $\{S_i\}_{i=1}^m$ covers $[n]$, then we can conclude that $q_{\mathrm{II}}(n,k) = q_{\mathrm{II}}(n,k)$. If not, assume, without loss of generality, that the last column of $M$ is the all-zero vector. By \cref{lemma_logbound}, $2^m \geq 2k > n$, so there exists a nonzero $v \in \mathbb{F}_2^m$ that is not a column of $M$. Our goal is to modify $M$ so that the last coupon is covered; this will be done by transferring $1$'s from the first $n-1$ columns to the last to turn it into $v$, while keeping the columns nonzero and distinct.

    Since this transfer only occurs in rows where $v_i = 1$, it is impossible for any of the first $n-1$ columns to become $v$. (Since they were already distinct from $v$, removing a $1$ from a position where $v_i = 1$ will not change that fact.) So it suffices to ensure that for all $i\in[m]$ where $v_i = 1$, we can select at least one column of weight greater than $1$ and whose $1$ in row $i$ can be transferred while keeping the first $n-1$ columns distinct. 

    For any such $i\in [m]$, the only way to create a duplicate column is to transfer a $1$ from a column that agrees with another column in every row except for row $i$, with the exception of the last column, as mentioned earlier. We claim that,
 for any $i\in [m]$, at least two columns of $M$ can have their $1$ in row $i$ transferred while keeping the columns distinct. Indeed, if this were not the case, then there would be at least $k-1$ columns with a $1$ in row $i$ matching a corresponding column that differs only in row $i$. This is impossible, since the number of columns (barring the last) with a $0$ in row $i$ is $n-k-1 < k-1$.

    Moreover, since the columns are distinct, at least one choice of column has a weight greater than $1$ for each $i\in[m]$ with $v_i =1$. Therefore, there is always at least one valid choice of a column for the transfer.

    Next, suppose $n \geq 2k-1$. It is clear that $q_\mathrm{I}(n+1,k) \leq q_{\mathrm{II}}(n,k)$, since a separating and covering system for $[n]$ is a separating system for $[n+1]$, as can be seen by adding an all-zeros column to the draw matrix. To prove the converse inequality, we start with a separating system for $[n+1]$ corresponding to $M \in \mathbb{F}_2^{m\times (n+1)}$. If $M$ has an all-zeros column, we can delete that column to obtain the desired result. Otherwise, assume that $M$ has no all-zeros column. The strategy is to transfer every $1$ from the last column while keeping the columns distinct, and then delete the resulting all-zeros last column. In each row, the difference between the number of $1$'s and $0$'s is $(n-(k-1)) - (k-1) = n-2k+2 \geq 1$, so using similar reasoning as before, we can show that there is always at least one valid choice of column to receive a $1$ from a given row in the last column.

     Finally, suppose $q_\mathrm{II}(n,k) > q_\mathrm{I}(n,k)$ for $n \geq 2k-1$. Let $m = q_{\mathrm{I}}(n,k)$ and $\{S_i \}_{i=1}^m \subseteq [n]$ be a separating system corresponding to draw matrix $M\in \mathbb{F}_2^{m\times n}$. By the hypothesis, $\{S_i\}_{i=1}^m$ does not cover $[n]$, so there will be exactly one all-zeros column in $M$, denoted as column $j$. Then, define $M'$ by inserting in $M$ a weight-$k$ row with a $1$ in the $j$-th column. Now, all columns of $M'$ are nonzero and distinct, so $q_{\mathrm{II}}(n,k) = q_{\mathrm{I}}(n,k) + 1$.

\end{document}